\documentclass{sigplanconf}

\usepackage{amsmath}
\usepackage{amssymb}
\usepackage{amsthm}
\usepackage[utf8]{inputenc}
\usepackage[english]{babel}
\usepackage{flushend}

\newtheorem*{theorem}{Theorem}

\begin{document}
\toappear{2015 Proceedings of the ACM Database Programming Languages Conference}

\setlength{\pdfpageheight}{\paperheight}
\setlength{\pdfpagewidth}{\paperwidth}

\title{The Gremlin Graph Traversal Machine and Language}

\authorinfo{Marko A. Rodriguez}
           {Director of Engineering at DataStax, Inc. \\ Project Committee Member of Apache TinkerPop}
           {marko@datastax.com}

\maketitle

\begin{abstract}
Gremlin is a graph traversal machine and language designed, developed, and distributed by the Apache TinkerPop project. Gremlin, as a graph traversal machine, is composed of three interacting components: a graph $G$, a traversal $\Psi$, and a set of traversers $T$. The traversers move about the graph according to the instructions specified in the traversal, where the result of the computation is the ultimate locations of all halted traversers. A Gremlin machine can be executed over any supporting graph computing system such as an OLTP graph database and/or an OLAP graph processor. Gremlin, as a graph traversal language, is a functional language implemented in the user's native programming language and is used to define the $\Psi$ of a Gremlin machine. This article provides a mathematical description of Gremlin and details its automaton and functional properties. These properties enable Gremlin to naturally support imperative and declarative querying, host language agnosticism, user-defined domain specific languages, an extensible compiler/optimizer, single- and multi-machine execution models, hybrid depth- and breadth-first evaluation, as well as the existence of a Universal Gremlin Machine and its respective entailments.
\end{abstract}



\section{Introduction}

A graph is a structure composed of vertices and edges. Graphs have seen a resurgence in the database community with the growth of graph database technology \cite{graph:angles2008}. The query language of a graph database typically promotes either a \textit{graph traversal} or a \textit{graph pattern match} perspective. In the traversal model, traversers walk a graph according to particular user provided instructions and the result of the traversal is the locations of all halted traversers. In the pattern match model, a subgraph containing variables is created by the user and all graph elements that bind to those variables are returned as the result set. Gremlin supports both the imperative traversal-style and the declarative pattern match-style within the same framework. Furthermore, beyond supporting both popular models of graph querying, Gremlin's machine and language structures naturally facilitate Gremlin being 1.) embedded in a host programming language, 2.) extended by users wishing to leverage the terminology of their problem domain, 3.) optimized via an extensible set of compile-time rewrite rules, 4.) executed within a multi-machine compute cluster, 5.) evaluated in a depth-first, breadth-first, or hybrid ordering, and finally, 6.) represented within the graph itself via the theoretical existence of a Universal Gremlin Machine.

\section{Graph Traversal Machine}

Gremlin, as a graph traversal machine, is composed of three components: a graph $G$ (data), a traversal $\Psi$ (instructions), and a set of traversers $T$ (read/write heads). Conceptually, a collection of traversers in $T$ move about $G$ according to the instructions specified in $\Psi$. The computation is complete when either 1.) there no longer exists any traversers in $T$ or 2.) all existing traversers no longer reference an instruction in $\Psi$ (i.e. they have halted). For the former, the result set is the empty set. For the latter, the result set is the multi-set union of the $G$ locations those halted traversers reference.

\subsection{The Graph}

Gremlin operates over a multi-relational, attributed, digraph $G = (V,E,\lambda)$, where $V$ is a set of vertices, $E \subseteq (V \times V)$ is a multi-set of directed binary edges, and $\lambda : ((V \cup E) \times \Sigma^*) \rightarrow \left(U \setminus (V \cup E)\right)$ is a partial function that maps an element/string pair to an object in the universal set $U$ (excluding vertices and edges as allowed property values). Given $\lambda$, every vertex and edge can have an arbitrary number of key/value pairs called \textit{properties}. For example, vertices may have name, age, latitude properties and edges may have weight, date, permission properties. The universal set $U$ contains the set of all property values. These values may be restricted to longs ($\mathbb{N}$), doubles ($\mathbb{R}$), strings ($\Sigma^*$), etc., or subsets thereof, and thus given the schema of the graph, $U$ can be constrained accordingly. 

\subsection{The Traversal}

A traversal $\Psi$ is a tree of functions called \textit{steps}. Steps are arranged in the following two ways:
\begin{enumerate}
  \item \textbf{Linear motif}: The traversal $f \circ g \circ h$ is a linear chain of three steps where the output traversers of $f$ are the input traversers of $g$. Likewise, the output of $g$ is the input of $h$.
  \item \textbf{Nested motif}: The traversal $f(g \circ h) \circ k$ contains the nested traversal $g \circ h$ which is an argument of the step $f$. In this way, $f$ will leverage $g \circ h$ in its mapping of its input traversers to its output traversers which are then provided as input to $k$.
\end{enumerate}
A step $f \in \Psi$ defined as $f : A^* \rightarrow B^*$ maps a set of traversers located at objects of type $A$ to a set of traversers located at objects of type $B$. The Kleene star notation $A^*$, when used in the context of an interpretation function, denotes that multiple traversers may be at the same element in $A$. However, what is mapped is a unique set of traversers to a unique set of traversers. It is only the multi-set union ($\uplus$) of their locations in $G$ that may contain duplicates.

The Gremlin graph traversal language defines approximately 30 steps which can be understood as the \textit{instruction set} of the Gremlin traversal machine. These steps are useful in practice, with typically only 10 or so of them being applied in the majority of cases. Each of the provided steps can be understood as being a specification of one of the 5 general types enumerated below.\footnote{If the underlying host language supports lambda functions (and \texttt{LambdaVerificationStrategy} is disabled), then it is possible for users to leverage the common lambda parameterization idiom of functional programming. For instance, users can do \texttt{filter\{t.loops() < 5\}}. However, this is strongly discouraged as the provided lambda can not be subjected to Gremlin's compiler optimizations. Instead, the pure traversal form \texttt{loops().is(lt(5))} should be used, where is() is a type of filter() step.}
\begin{enumerate}
  \item $\textbf{map} : A^* \rightarrow B^*$, where $|\text{map}(T)| = |T|$. These functions map the traversers $T$ at objects of type $A$ to a set of traversers at objects of type $B$ without altering the traverser set size.
  \item $\textbf{flatMap} : A^* \rightarrow B^*$, where the output traverser set may be smaller, equal to, or larger than the input traverser set.
  \item $\textbf{filter} : A^* \rightarrow A^*$, where $\text{filter}(T) \subseteq T$. The traversers in the input set are either retained or removed from the output set.
  \item $\textbf{sideEffect} : A^* \rightarrow_x A^*$, where $\text{sideEffect}(T) = T$ . An identity function operates on the traversers though some data structure $x$ (typically in $G$) is mutated in some way.
  \item $\textbf{branch} : A^* \rightarrow^b B^*$, where an internal branch function $b : T \rightarrow \mathcal{P}(\Psi)$ maps a traverser to any number of the nested traversals' start steps.
\end{enumerate}

All the above steps can be represented as a specification of flatMap() -- i.e. map one set of traversers to another set of traversers. For instance, map() as a flatMap() simply maps each traverser in $T$ (at $A$) to a single traverser in $B$. For filter(), flatMap() either includes the original input set traverser or removes it, where $A = B$. If flatMap() can mutate outside data structures, then sideEffect() is simulated, where the input traverser set $T$ is the output traverser set. Finally, branch() is simulated by ensuring the internal logic of flatMap() include rules for choosing different mappings of the traversers in $T$ given their state.

\subsection{The Traverser}

A traverser unifies the graph and the traversal via a reference to an object in the graph and a reference to a step in the traversal. Formally, a traverser $t$ is an element in the 6-tuple set
\begin{equation*}
T \subseteq \left(U \times \Psi \times \left(\mathcal{P}(\Sigma^*) \times U\right)^* \times \mathbb{N}^+ \times U \times \mathbb{N}^+\right).
\end{equation*}
The first element of the tuple is the traverser's location in the graph $G$ (e.g. $v \in V$, where $V \subset U$).\footnote{The ``graph location" of a traverser is in $U$ as opposed to only $V \cup E$ because a traverser can move beyond vertices and edges by referencing arbitrary objects associated with the graph such as property keys, property values, and side-effect data structures.} The second element is the traverser's step location in the traversal $\Psi$. The third element is a sequence of sets of strings and objects called a \textit{labeled path}. For example, $(((a),x),((b,c),y),(\emptyset,z))$ denotes the traverser's path $x \leadsto y \leadsto z$ with respective step labels at each location. The forth element is the traverser's \textit{bulk} which denotes the number of equivalent traversers this particular traverser represents.\footnote{Traverser bulk is useful as an fundamental optimization, though it is not theoretically required.} The firth element is the traverser's \textit{sack} which is a local variable of the traverser. The sixth, and final element, is the traverser's \textit{loop counter} which specifies the number of times a traverser has gone through a loop sequence. The following functions project the aforementioned components of a traverser to their respective values.
\begin{enumerate}
  \item $\mu : T \rightarrow U$ maps a traverser to an object in $U$ (i.e. its location in the graph).
  \item $\psi : T \rightarrow \Psi$ maps a traverser to a step in $\Psi$ (i.e. its location in the traversal -- program counter).
  \item $\Delta : T \rightarrow (\mathcal{P}(\Sigma^*) \times U)^*$ maps a traverser to its labeled path (i.e. its history in the graph).
  \item $\beta : T \rightarrow \mathbb{N}^+$ maps a traverser to its bulk.
  \item $\varsigma : T \rightarrow U$ maps a traverser to its sack value.
  \item $\iota : T \rightarrow \mathbb{N}^+$ maps a traverser to its loop counter.
\end{enumerate}
Visually, a traverser $t \in T$ is a ``bundle" of local variables (metadata) with a projection to a location in the graph $G$ and a projection to a location in the traversal $\Psi$.
\begin{equation*}
G \longleftarrow\text{}_\mu \;\; \frac{t \in T}{\{\Delta,\beta,\varsigma,\iota\}} \;\; \text{}_\psi\longrightarrow \Psi
\end{equation*}

\section{Graph Traversal Language}

Gremlin, as a graph traversal language, is a \textit{functional language}. The purpose of the language is to enable a human user to easily define $\Psi$ and thus, program a Gremlin machine. The simplicity of Gremlin's grammar enables it to be embedded in the native programming language of the user.\footnote{TinkerPop distributes a Gremlin machine implemented in Java8 and a Gremlin language binding in both Java8 and Groovy. With the JVM being host to numerous programming languages, the wider TinkerPop community has provided Gremlin language bindings in Scala, Clojure, Ruby, JavaScript, and more. Conceptually, TinkerPop's Gremlin machine is a virtual machine implemented in Java that can be programmed via the numerous Java-based programming languages in existence.} In this way, for a developer, there is no discontinuity between their software code and their graph analysis code. 

In order for a language to host Gremlin, the language needs to support \textit{function composition} and \textit{function types} (i.e. functions as first-class entities or enable it via ``function objects."). With method chaining (a type of function composition), a natural fluent syntax is possible. For instance, the traversal $a \circ b \circ c$ is denoted \texttt{a().b().c()} in the dot notation-syntax of modern object-oriented programming languages. With function arguments, traversal nesting is possible. For instance, $a(b \circ c) \circ d$ is denoted \texttt{a(b().c()).d()}.

\subsection{A Simple Traversal}

The most basic graph traversal is one that moves traversers through the steps of $\Psi$ in a sequential order ($\Psi_1 \leadsto \Psi_2 \leadsto \ldots \leadsto \Psi_{|\Psi|}$) and where no step maintains internal, nested traversals. The example traversal below is a simple linear traversal that determines, in plain language, the age of the oldest person that Marko knows (assuming, for the sake of simplicity of discussion, that each vertex in the example graph has a unique name). 
\begin{verbatim}
  g.V().has("name","marko").
    out("knows").values("age").max()
\end{verbatim}
The first term, $V_g$ (\texttt{g.V()}), is the definition of a traverser set bijective to $V$, where $\biguplus_i \mu((V_g)_i) = V$. The above traversal can be written in curried functional notation as 
\begin{equation*}
\text{max}(\text{values}_\text{age}(\text{out}_\text{knows}(\text{has}_\text{name=marko}(V_g)))).
\end{equation*}
The starting traverser set $V_g$ is first processed by $\text{has}_\text{name=marko}$. A traverser set is returned that only contains a single traverser at the vertex named ``marko." The step $\text{out}_\text{knows}$ then maps the \textit{marko}-vertex traverser (parent) to a traverser set (children) located at those vertices that are outgoing \textit{knows}-adjacent to the \textit{marko}-vertex. The children have a new graph location, traversal step location, and a path that is the concatenation of their parent's path and their current location. An example child traverser, at this point, will be the 6-tuple
\begin{equation*}
(y, \text{values}_\text{age}, ((\emptyset,x),(\emptyset,y)), 1, \emptyset, 0),
\end{equation*}
where $\lambda(x,\text{name}) = \text{marko}$ and $\lambda((x,y),\text{label}) = \text{knows}$. Next, $\text{values}_\text{age}$ maps to a traverser set where each child traverser is located at the integer value of their current vertex's \textit{age}-property. Finally, $\text{max}()$ transforms the traverser's at $\mathbb{N}^*$ to a single traverser at a number representing the maximum number in the previous set (i.e. the oldest \textit{age}-value).
\begin{equation*}
\begin{array}{rlcr}
V_g &: \mathbf{0} \rightarrow V^* & & \text{flatMap} \\
\text{has}_\text{name=marko} &: V^* \rightarrow V^* & & \text{filter} \\ 
\text{out}_\text{knows} &: V^* \rightarrow V^* & & \text{flatMap} \\
\text{values}_\text{age} &: V^* \rightarrow \mathbb{N}^* & & \text{map} \\
\text{max} &: \left[\mathbb{N}^*\right] \rightarrow \mathbb{N} & & \text{map}
\end{array}
\end{equation*}
It is important to note that the domain of the $step_n$ is equal to the range of $step_{n-1}$. Furthermore, the domain of max() is $[N^*]$ and the range is $N$. The step max() is ``blocking" in that the entire traverser set is required as input before the single traverser $t$ is outputted, where $\mu(t) \in \mathbb{N}$. The notation $[A^*]$ denotes a barrier. Steps that reduce a traverser set to a single traverser by way of some binary operation are called \textit{reducing barrier steps}. 

In Gremlin, a traverser set can grow and shrink over the course of the computation. Traverser sets typically shrink due to filter() steps removing traversers, map()/flatMap() partial functions mapping to undefined locations in $G$ (e.g. Marko may not know anyone), or reducing barrier steps going from many-to-one.\footnote{As will be explain later, traverser sets also shrink when multiple traversers arrive at the same location in $G$ and, at which point, these traversers merge into a single traverser with a respective ``bulk" equal to the sum of the bulks of all merged traversers. However, while the set shrinks, the same logical number of traversers still exists.} Traverser sets grow due to one-to-many flatMap() steps. Gremlin traversers are \textit{furcating automata} in that if multiple options are met, then all options are taken. For instance, if a traverser is at a single vertex and the current step $\psi(t)$ maps to many adjacent vertices (e.g. $\text{out}_\text{knows}$ and Marko knows many people), then the traverser ``splits" (clones itself) and each child is placed at each adjacent vertex. The only modification to the child clones are new locations in $G$ and $\Psi$ as well as a new labeled path $\Delta$ which is a one-step extension of the parent traverser's path.

The language used in the discussion thus far states that a ``set of traversers" is being mapped between each step of the traversal. However, traversers are isolated entities maintaining their own metadata/state, where the step functions themselves have no state. This type of traverser isolation enables a traversal's evaluation order to change at different points in $\Psi$ as sometimes it is useful to use depth-first (one traverser at each step) and sometimes breadth-first (sets of traversers at each step). A Gremlin machine implementation can make use of a dataflow/stream construct \cite{dataflow:lee1995} and simulate breadth-first evaluation at particular points in the traverser stream via the insertion of an \textit{identity barrier step} with interpretation function $\text{barrier} : \left[U^*\right] \rightarrow U^*$.

\subsection{A Branching Traversal}

A branch in Gremlin is a split in $\Psi$. Formally,
\begin{equation*}
\text{branch}_\text{b}(t) = b(t)(t),
\end{equation*}
where $b : T \rightarrow \Psi$. The branching function $b$ determines, given the state of $t$, which internal traversal the traverser should follow. Depending on the particular branch step, the traverser $t$ may be sent down a single branch (e.g. choose()), a subset of the branches (e.g. repeat().emit()), or all branches (e.g. union()).

The choose() step is a branch step which provides the common ``if/else if/.../else" programming construct.
\begin{verbatim}
  g.V().choose(label()).
    option("person", out("created").count()).
    option("software", in("created").count()).
    option(none, label())
\end{verbatim}  
In the above traversal, if the traverser incoming to choose() is at a \textit{person}-vertex, then send the traverser down the branch that computes the number of projects that that person has created. If the traverser is at a \textit{software}-vertex, then send the traverser down the branch that computes the number of collaborators on that software project. Finally, if the traverser is located at neither a \textit{person}- nor \textit{software}-vertex, then send the traverser down the branch that yields the label of the vertex, where the \textit{none}-option refers to a branch that should be taken if no other options are valid (i.e. ``else"). If the \textit{none}-option did not exist, then choose() would act as a filter removing the option-less traverser from $T$. Note that option() is not a step, but a \textit{step modulator}. Step modulators are ``syntactic sugars" that manipulate the previous step in order to reduce the complexity of the modulated step's arguments (and respective function overloadings). In this way, choose() takes traversals as arguments and thus, maintains internal nested traversals, where the first traversal (label()) plus the option keys (e.g. ``person") form the branch function. The above choose() step is represented in curried form as
\begin{equation*}
\text{choose}_{\text{label}}(t) =
  \begin{cases}
    \text{count}(\text{out}_\text{created}(t)) &: \mu(\text{label}(t)) = \text{person} \\
    \text{count}(\text{in}_\text{created}(t)) &: \mu(\text{label}(t)) = \text{software} \\   
    \text{label}(t) &: \text{otherwise}.
  \end{cases}
\end{equation*}
Note that the domain and range of $\text{choose}_\text{label}()$ is
\begin{equation*}
\text{choose}_{\text{label}} : V^* \rightarrow \left(\mathbb{N}^+ \cup \Sigma^*\right)^*.
\end{equation*}
As such, any step following choose() must be able to accept either numbers or strings.

It is worth noting that Gremlin supports a more compact syntax for boolean-based ``if/else." If there are only two options, ``person" and \textit{none}, then the above traversal would be defined as below.
\begin{verbatim}
  g.V().choose(label().is("person"), 
    out("created").count(),
    label())
\end{verbatim}  

\subsection{A Recursive Traversal}

The examples presented thus far have the traversers moving from ``left to right" through the sequence of steps in $\Psi$. In order to support recursion (i.e. looping), it is necessary to set the traverser's $\psi$-program counter back to some previously seen step. An example of such a step is the recursive function
\begin{equation*}
\text{repeat}_p(t) =
\begin{cases}
\text{repeat}_p(t_{\iota+1}) &: p(t) = \textbf{true} \\
t_0 &: \text{otherwise},
\end{cases}
\end{equation*}
where $p : T \rightarrow \textbf{Bool}$ is some traverser predicate, $\iota(t_{\iota+1}) = \iota(t) + 1$ (i.e. increment the loop counter), and $\iota(t_0) = 0$ (i.e. reset the loop counter).

The following traversal returns the names of the vertices 5 outgoing steps from the vertex named ``marko."
\begin{verbatim}
  g.V().has("name","marko").
    repeat(out()).times(5).
      values("name")
\end{verbatim}
With times() being a step modulator, the repeat() step is functionally defined as
\begin{equation*}
\text{repeat}_{\iota < 5}(t) =
\begin{cases}
 \text{repeat}_{\iota < 5}(\text{out}(t_{\iota+1})) &:  \iota(t) < 5 \\
 t_{0} &: \text{otherwise}.
\end{cases}
\end{equation*}

Suppose it is necessary to get the names of all the vertices encountered along the 5 step walk emanating from the vertex named ``marko" (and not just those names 5 steps away).
\begin{verbatim}
  g.V().has("name","marko").
    repeat(out()).emit().times(5).
      values("name")
\end{verbatim}
If the traverser loops, it is also emitted along with its recursive mapping. Note that emit(), like times(), is a step modulator of repeat(), where
\begin{equation*}
\text{repeat}_{\iota < 5,\text{emit}}(t) =
\begin{cases}
\left( \text{repeat}_{\iota < 5,\text{emit}}(\text{out}(t_{\iota+1})), t_0\right) &: \iota(t) < 5 \\
t_0 &: \text{otherwise}.
\end{cases}
\end{equation*}

\subsection{A Path Traversal}

The third component of the traverser $t$'s 6-tuple is its labeled path $\Delta(t) \in \left(\mathcal{P}(\Sigma^*) \times U\right)^*$. Whenever a traverser is mapped to a new location in $G$, this location as well as the set of labels for the respective step in $\Psi$ are appended to the child traverser's path. For example, assume the following traversal.
\begin{verbatim}
  g.V().as("a").out().as("b","c").path()
\end{verbatim}
In the traversal above, the traverser $t$ will start at a particular vertex in $x \in V$. That location is labeled ``a" via the step modulator $\text{as}()$, where $\Delta(t) = (((a),x))$.  Next, the traverser $t$ will split itself amongst all the outgoing adjacent vertices of $x$, where one particular child traverser's path would be $\Delta(t') = (((a),x),((b,c),y))$ assuming $(x,y) \in E$. Thus,
\begin{equation*}
\text{path}(t) = t_{\Delta(t)},
\end{equation*}
where $\mu\left(t_{\Delta(t)}\right) = \Delta(t)$. A single halted traverser $t''$ from the traversal above would have
\begin{equation*}
\mu(t'') = (((a),x),((b,c),y))
\end{equation*} 
and
\begin{equation*}
\Delta(t'') = (((a),x),((b,c),y),(\emptyset, ((a),x),((b,c),y)))).
\end{equation*}
That is, the labeled path of $t''$, up to that point in the traversal, is an element in its path.

A traverser's path history is useful in the following enumerated situations.
\begin{enumerate}
  \item It is necessary to determine the (shortest)-path from vertex $x$ to vertex $y$.
  \item It is necessary to go back to some previous location of the traverser.
  \item It is necessary to determine if a particular location has already been visited.
\end{enumerate}
In terms of items 1 and 3, 
\begin{verbatim}
  g.V(x).repeat(out().simplePath()).
    until(is(y)).path().limit(1)
\end{verbatim}
will return the shortest, simple (non-looping) path from vertex $x$ to vertex $y$, where until() is a step modulator for repeat() and the filter
\begin{equation*}
\text{simplePath}(t) =
  \begin{cases}
    t &: \left|\bigcup_{i < |\Delta(t)|}{\mu\left(\Delta(t)_i\right)}\right| = \left|\Delta(t)\right| \\   
    \emptyset &: \text{otherwise}.
  \end{cases} 
\end{equation*}

\subsection{A Projecting Traversal}

In the previous subsection, it was stated that sometimes it is necessary to go back to some previous location in the traverser's path history. The following traversal does just that.
\begin{verbatim}
  g.V().as("a").out("knows").as("b").
    select("a","b").
      by(in("knows").count()).
      by(out("knows").count())
\end{verbatim}
When the traverser $t$ reaches select(), there will be two vertices labeled ``a" and ``b" in its path. The select() step generates two new traversers $t_{\Delta_a(t)}$ and $t_{\Delta_b(t)}$, where $\mu\left(t_{\Delta_a(t)}\right) = \Delta_a(t)$  and $\mu\left(t_{\Delta_b(t)}\right) = \Delta_b(t)$. Traverser $t_{\Delta_a(t)}$ will ultimately determine the number of incoming \textit{knows}-adjacent vertices to the ``a"-vertex and traverser $t_{\Delta_b(t)}$ will determine the number of outgoing \textit{knows}-adjacent vertices to the ``b"-vertex. The $\text{by}()$ step modulator specifies which traversal the ``a" and ``b" traversers should traverse. The curried function signature is
\begin{equation*}
\text{select}_\text{a,b} : V^* \rightarrow \mathcal{P}(\Sigma^* \times \mathbb{N}^+)^*, 
\end{equation*}
where an element in $\mathcal{P}(\Sigma^* \times \mathbb{N}^+)$ is a \texttt{Map<String,Long>} data structure in programming. The definition of the $\text{select}_\text{a,b}$ function is
\begin{equation*}
\text{select}_\text{a,b}(t) =
\left(
\begin{array}{c}
\left(a,\text{count}\left(\text{in}_\text{knows}\left(t_{\Delta_a(t)}\right)\right)\right), \\
\left(b,\text{count}\left(\text{out}_\text{knows}\left(t_{\Delta_b(t)}\right)\right)\right)
\end{array}
\right).
\end{equation*}

The $\text{where}()$ step is similar to $\text{select}()$ save that it filters a traverser based on its labeled path. The traversal below does the same selection as above, but only if the traverser's ``a" and ``b" vertices are not maternal siblings. Thus,
\begin{equation*}
\neg\text{where}_\text{a,b}(t) =
  \begin{cases}
    t &: t_{\Delta_b(t)} \notin \text{in}_\text{mother}\left(\text{out}_\text{mother}\left(t_{\Delta_a(t)}\right)\right) \\
    \emptyset &: \text{otherwise}
  \end{cases}
\end{equation*}
in the traversal
\begin{verbatim}
  g.V().as("a").out("knows").as("b").
    where(not(
      as("a").out("mother").in("mother").as("b"))).
    select("a","b").
      by(in("knows").count()).
      by(out("knows").count())
\end{verbatim}
The above syntax of \texttt{as("a")...as("b")} is syntactic sugar for \texttt{select("a")...where(eq("b"))}.

\subsection{A Centrality Traversal}

A graph is a complex structure that is difficult to reason about in its $n$-dimensional form. In order to extract meaningful information from graphs, numerous statistics have been developed in the domains of graph theory and network science \cite{netanal:brandes2005}. Every graph statistic maps an $n$-dimensional graph to some lower dimensional space. Typically, the reduction is either to a 0- or 1-dimensional space. For instance, $\text{size}_\text{V} : \mathbb{G} \rightarrow \mathbb{N}^+$ is a 0-dimensional statistic that maps a graph to the number of vertices it contains. \textit{Graph centrality} is a 1-dimensional graph statistic generally defined as $\mathbb{G} \rightarrow \mathbb{R}^{|V|}$, where the function maps a graph $G$ to a vector of centrality scores for each vertex in $V$. Centrality measures identify ``important," ``representative," ``connective," ``influential," etc. vertices in the graph. In \textit{eigenvector centrality}, centrality is formally defined as the probability that a vertex will be host to some random walker at some random point in time. This description can be represented by the linear algebraic equation $A\mathbf{v} = \lambda\mathbf{v}$, where
\begin{equation*}
  A_{i,j} = 
  \begin{cases}
    1 & : (i,j) \in E \\
    0 & : \text{otherwise}
  \end{cases}
\end{equation*}
is the adjacency matrix of the graph $G$ and $\mathbf{v}$ is the eigenvector whose components change, with each iteration, according to the scalar $\lambda$ (the eigenvalue).\footnote{The method of multiplying the vector $\mathbf{v}$ (initially being set to $\mathbf{1}^{|V|}$) against the adjacency matrix $A$ until $\lambda$ reaches a fixed point is called the \textit{power iteration method}. The resultant $\mathbf{v}$ (when $\lambda$ converges) is guaranteed to be the primary eigenvector with $\lambda$ being the largest eigenvalue. $\lambda$ defines the growth rate of $\mathbf{v}$ over any subsequent iterations. In Apache TinkerPop's Gremlin implementation, the growth rate $\lambda$ plays an important role in understanding when the ``bulk" $\beta$ of a traverser will overflow its 64-bit long representation.}  If the graph is strongly connected and aperiodic \cite{markov:haggstrom2002}, then the larger a vertex's value in $\mathbf{v}$, the more central it is in the graph. The equation $A\mathbf{v} = \lambda\mathbf{v}$ is expressed and solved in Gremlin via the traversal
\begin{verbatim}
  g.V().repeat(groupCount("m").out()).
    times(30).cap("m")
\end{verbatim}
The groupCount() side-effect step is defined as
\begin{equation*}
\text{groupCount}_\text{m}(t) =_{\text{m}} t : m[\mu(t)] = m[\mu(t)] + \beta(t).
\end{equation*}
Every time a traverser arrives at groupCount(``m"), its vertex location is indexed into the map $m \in \mathcal{P}(V \times \mathbb{N}^+)$, where $m \sim \mathbf{v}$ and $m$ is a  \texttt{Map<Vertex,Long>}. The keys of $m$ are the vertices in $V$ and the values are the number of times each vertex has been encountered thus far. Thirty iterations is provided as a value that is typically large enough to ensure convergence in natural graphs. It is possible for repeat()'s times() to be replaced with an until() that calculates whether the map's values have reached a steady state distribution. This requires comparing the unit vector $\hat{\mathbf{v}} = \frac{\mathbf{v}}{\|\mathbf{v}\|}$ at iteration $n$ with $\hat{\mathbf{v}}$ at iteration $n+1$. However, for the sake of simplicity, this computation is left to the reader to deduce.

The step \text{cap(``m")} is a \textit{supplying barrier step} in that is takes takes all the incoming traversers and emits a single traverser that is not a function of the incoming set, but a function of a side-effect data structure. In this case, a single traverser $t$ is emitted where $\mu(t) = m$.
\begin{equation*}
\text{cap}_\text{m} : [V^*] \rightarrow \mathcal{P}(V \times \mathbb{N}^+).
\end{equation*}

\subsection{A Mutating Traversal}

All of the examples presented thus far have only read from the graph. None have written to it. Gremlin provides a collection of graph mutation steps that can be used to add and remove vertices, edges, and properties. A few of these are outlined below.
\begin{equation*}
\begin{array}{rlcr}
\text{addOutE} &: V^* \rightarrow E^* & & \text{sideEffect/map} \\ 
\text{addInE} &: V^* \rightarrow E^* & & \text{sideEffect/map} \\
\text{addV} &: U^* \rightarrow V^* & & \text{sideEffect/map} \\
\text{property} &: \left(V \cup E\right)^* \rightarrow \left(V \cup E\right)^* & & \text{sideEffect} \\
\text{drop} &: \left[\left(V \cup E\right)^*\right] \rightarrow \mathbf{0} & & \text{sideEffect/map} \\
\end{array}
\end{equation*}

The two traversals below mutate the graph. The first one adds an inverse \textit{createdBy}-edge for every \textit{created}-edge. The second removes the original \textit{created}-edges.\begin{verbatim}
  g.V().as("a").out("created").
    addOutE("createdBy","a")
  
  g.V().outE("created").drop()
\end{verbatim}

\subsection{A Declarative Traversal}

Gremlin supports graph pattern matching analogous, in many respects, to SPARQL \cite{sparql:prud2004}. The primary benefit of Gremlin's pattern matching is that it encompasses only a single step within the Gremlin language and thus, it is possible to move from declarative pattern matching, to imperative traversals, and back all within the same traversal definition. Moreover, pattern matching is expressed as a traversal and thus, uses the same Gremlin traversal machine constructs presented thus far. 

The match() step's argument is a set of \textit{traversal patterns} that may be not()'d or nested via and() and or(). When a traverser enters match(), it will propagate through each pattern. A traverser that continues on to the step after match() is guaranteed to have had its labeled path values bind (via equality) to all the prefix and postfix variables of the match() traversal patterns.
\begin{verbatim}
  g.V().match(
      as("a").out("created").as("b"),
      as("b").in("created").count().is(gt(3)),
      as("b").in("created").as("c"),
      as("a").out("father").as("c")).
    dedup("a").
    select("a").by("name")
\end{verbatim}
The traversal above will return the name of all vertices who created a piece of software in collaboration with at least 4 people with the caveat that one of those collaborators is their father.
\newline\newline
$\text{match}(t) =$
\begin{equation*}
\begin{cases}
\text{match}(\text{bind}_\text{b}(\text{out}_\text{created}(t_{\Delta_a(t) \wedge \Delta_\text{m1}}))) &: \Delta_a(t) \neq \emptyset = \Delta_{m1}(t) \\
\text{match}(\text{is}_{> 3}(\text{count}(\text{in}_\text{created}(t_{\Delta_b(t) \wedge \Delta_\text{m2}})))) &: \Delta_b(t) \neq \emptyset = \Delta_{m2}(t) \\
\text{match}(\text{bind}_\text{c}(\text{in}_\text{created}(t_{\Delta_b(t) \wedge \Delta_\text{m3}}))) &: \Delta_b(t) \neq \emptyset = \Delta_{m3}(t) \\
\text{match}(\text{bind}_\text{c}(\text{out}_\text{father}(t_{\Delta_a(t) \wedge \Delta_\text{m4}}))) &: \Delta_a(t) \neq \emptyset = \Delta_{m4}(t) \\
t &: \text{otherwise},
\end{cases}
\end{equation*}
where
\begin{equation*}
\text{bind}_\text{x}(t) =
\begin{cases}
t_{\Delta_x(t) = \mu(t)} &: \Delta_x(t) = \emptyset \\
t &: \Delta_x(t) = \mu(t) \\
\emptyset &: \text{otherwise}.
\end{cases}
\end{equation*}
The $\text{match}()$ step is a recursively defined branch step, where each pattern/branch is taken once and only once. This is guaranteed by the hidden path label ``m\#" which is appended to the traverser's labeled path upon entering a branch pattern. Moreover, the prefix label $x$ of a pattern must exist in the traverser's path prior to the traverser taking that particular branch (i.e. $\Delta_x(t)$ can not equal $\emptyset$). Upon completing a pattern, if the traverser already has the postfix label in $\Delta(t)$ then that historic location must equal its current location $\mu(t)$ otherwise the traverser is deemed a non-match and is filtered out. However, if the postfix label does not exist in $\Delta(t)$, then it is added to the path of the traverser and thus, that variable is bound for all subsequent patterns. A traverser is able to exit the match() step when every pattern has been taken. Within the labeled path of a surviving traverser lies the match-variables and their respective bindings -- $\left(\Delta_a(t), \Delta_b(t), \Delta_c(t)\right)$.

The order in which match-patterns are executed is up to the $\text{match}()$ step implementation. The only caveat is that the two pattern selection criteria are respected -- 1.) the pattern has not been taken before and 2.) the prefix variable of the pattern already exists in the labeled path of the traverser. TinkerPop's Gremlin implementation provides two \textit{match algorithms} called \texttt{GreedyMatchAlgorithm} and \texttt{CountMatchAlgorithm} \cite{budget:broecheler2011}. The former simply finds the first pattern in the user provided list that meets the respective constraints and executes that pattern. The latter maintains, for each pattern, a dynamic \textit{multiplicity} variable that is equal to the number of traversers outputted by the pattern divided by the number of traversers inputted to the pattern. It then continually re-sorts the patterns favoring those that have the lowest multiplicity (i.e. it favors patterns that yield the largest set reductions).

\subsection{A Domain Specific Traversal}

The Gremlin traversal machine supports approximately 30 steps in its instruction set. These steps are deemed the most useful for most any traversal algorithm. The Gremlin traversal language is (nearly) in one-to-one correspondence with these steps. This alignment is due to the fact that the Gremlin language is a ``graph specific language" that forces the user to process their data from the graph-perspective of vertices (out(), inV()), edges (outE()), and properties (values()).

Typically, in practice, a graph structure represents some problem-domain that is best modeled as a graph. In these domains, the concept of vertices and edges may be understood as people and social relationships, for example. It is trivial for a user to define a \textit{domain specific language} that, when compiled, generates a traversal in terms of the $\sim$30 Gremlin steps. For instance, given the example graph used throughout this section, a hypothetical ``social traversal language" may allow the following domain specific query to be expressed.
\begin{verbatim}
  g.people().named("marko").
    who().know(well).people().
    who().created("software").
    are().named()
\end{verbatim}
Each one of these steps would be the composite or 1 or more Gremlin steps.
\begin{equation*}
\begin{array}{lcl}
\texttt{g.people()} &\mapsto& \texttt{g.V().hasLabel("person")} \\
\texttt{named("marko")} &\mapsto& \texttt{has("name","marko")} \\ 
\texttt{who()} &\mapsto& \texttt{identity()} \\
\texttt{know(well)} &\mapsto& \texttt{outE("knows").} \\
& & \;\;\texttt{has("weight",gt(0.75)).} \\
& & \;\;\texttt{inV()} \\
\texttt{created("software")} &\mapsto& \texttt{out("created")}. \\
& & \;\;\texttt{hasLabel("software")} \\
\texttt{are()} &\mapsto& \texttt{identity()} \\ 
\texttt{named()} &\mapsto& \texttt{values("name")} 
\end{array}
\end{equation*}
While the user expresses queries in the language of their domain, Gremlin ultimately evaluates those queries in terms of the underlying graph structure used to represent that domain.

\section{Traversal Strategies}

A traverser executes the step functions defined in $\Psi$. Sometimes a particular step sequence in $\Psi$ can be expressed in another way that is perhaps more efficient to execute. Gremlin \textit{traversal strategies} define translations which rewrite sections of a traversal (with typically, though not necessarily, the same semantics as the original traversal) \cite{pathalg:rodriguez2009}. There are 5 traversal strategy categories that form a total order, where a category's strategies are evaluated prior to moving to the next category and prior to traversal execution. Furthermore, within a category, strategies form a partial order where some strategies may require the execution of another strategy before (or after) its execution.
\begin{enumerate}
 \item \textbf{Decoration}: Rewrite a traversal given that certain steps serve only as syntactic placeholders.
 \item \textbf{Optimization}: Rewrite a traversal given that a particular step sequence can be expressed in a more efficient form.
 \item \textbf{Vendor Optimization}: Rewrite a traversal give that a particular step sequence can be expressed in a more efficient form given the underlying graph system.
 \item \textbf{Finalization}: Make any final adjustments to the traversal given the ``final" compiled form.
 \item \textbf{Verification}: Analyze the traversal and ensure it is valid given some set of constraints.
\end{enumerate}

The itemization below presents a subset of Gremlin's provided traversal strategies and demonstrates a few rewrite rule examples for each.\footnote{It is easy for vendors and users to register new traversal strategies with TinkerPop's Gremlin compiler. This is typically done by vendors as their graph system may maintain different optimizations accessible only through their custom interfaces. In such situations, vendors will write custom steps and traversal strategies that replace particular Gremlin step sequences with their vendor-specific steps.} The rewrite rules are represented such that when the left-hand pattern is matched, it is rewritten as the right-hand pattern.
\begin{itemize}
\item \textbf{ConjunctionStrategy} (Decoration): Gremlin supports prefix and infix notation for logical connectors. When infix notation is used, it is converted into the respective prefix-based representation.
\begin{equation*}
\begin{array}{lcl}
\texttt{a.and().b} &\mapsto& \texttt{and(a,b)} \\
\texttt{a.or().b} &\mapsto& \texttt{or(a,b)} \\ 
\texttt{a.or().b.and().c} &\mapsto& \texttt{or(a,and(b,c))} \\
\texttt{a.and().b.or().c} &\mapsto& \texttt{or(and(a,b),c)}
\end{array}
\end{equation*}
\item \textbf{IncidentToAdjacentStrategy} (Optimization): If a traversal touches the incident edges of a vertex on its way to its adjacent vertices, and executes no step that requires the analysis of those incident edges, then simply jump to the adjacent vertices without manifesting the respective incident edges.
\begin{equation*}
\begin{array}{lcl}
\texttt{a.outE().inV().b} &\mapsto& \texttt{a.out().b} \\
\texttt{a.bothE().otherV().b} &\mapsto& \texttt{a.both().b}
\end{array}
\end{equation*}
\item \textbf{AdjacentToIncidentStrategy} (Optimization): If a traversal only checks for the existence of adjacent vertices, it is typically cheaper to only manifest the incident edges.\footnote{This is especially true in distributed Gremlin where a vertex in a vertex partition maintains direct references to only its properties, incident edges and their properties. Without \texttt{AdjacentToIncidentStrategy}, the left-handed patterns would waste network bandwidth as traversers would be sent to adjacent vertices only to be counted/etc.}  
\begin{equation*}
\begin{array}{lcl}
\texttt{a.in().count().b} &\mapsto& \texttt{a.inE().count().b} \\
\texttt{a.where(out()).b} &\mapsto& \texttt{a.where(outE()).b} \\
\texttt{a.and(in(),out()).b} &\mapsto& \texttt{a.and(inE(),outE()).b}
\end{array}
\end{equation*}
\item \textbf{IdentityRemovalStrategy} (Optimization): If a traversal maintains an identity() step, remove the identity() step and fold any as() modulators into the previous step.
\begin{equation*}
\begin{array}{lcl}
\texttt{a.identity().b} &\mapsto& \texttt{a.b}
\end{array}
\end{equation*}
\item \textbf{FilterRankingStrategy} (Optimization): All filter() steps either remove or retain traversers in $T$. Thus, if there is a sequence of filter() steps, reorder them such that cheaper filters are executed first in the hopes of yielding large set reductions prior to executing more costly filters. Note that order() is considered a ``filter" even though it only sorts $T$ without ever removing traversers from $T$.
\begin{equation*}
\begin{array}{lcl}
\texttt{a.and(c,d).has().b} &\mapsto& \texttt{a.has().and(c,d).b} \\
\texttt{a.order().dedup().b} &\mapsto& \texttt{a.dedup().order().b}
\end{array}
\end{equation*}
\item \textbf{RangeByIsCountStrategy} (Optimization): If a traversal only needs to determine if a particular number of elements exist, then instead of counting the full set, limit the count to 1 plus the required number.
\begin{equation*}
\begin{array}{lcl}
\texttt{a.count().is(0)} &\mapsto& \texttt{a.limit(1).count().is(0)} 
\end{array}
\end{equation*}
\item \textbf{XGraphStepStrategy} (Vendor Optimization): Graph database vendors typically maintain indices over the properties of the vertices (and sometimes edges) of the graph. In order avoid $\mathcal{O}(|V|)$ linear costs, fold all has() steps into a vendor specific V() step in order to facilitate $\mathcal{O}(log(|V|))$ indexed-based lookups.
\begin{equation*}
\begin{array}{lcl}
\texttt{V.has().has().b} &\mapsto& \texttt{V[has,has].b} 
\end{array}
\end{equation*}
\item \textbf{MatchPredicateStrategy} (Optimization): If match() is followed by a where() step, fold that where() step into match() as a new branch pattern. In this way, the where()-based pattern is subject to the match() step's \texttt{XMatchAlgorithm} runtime optimizer. Furthermore, if match() maintains a has()-based pattern whose prefix variable is the input variable, then pull that pattern out of match() such that it may be used by the vendor for respective index lookups (see \textbf{XGraphStepStrategy}).
\begin{equation*}
\begin{array}{lcl}
\texttt{a.match(c,d).where(e).b} &\mapsto& \texttt{a.match(c,d,e)} \\
\texttt{a.match(has(),c,d).b} &\mapsto& \texttt{a.has().match(c,d).b}
\end{array}
\end{equation*}
\item \textbf{ProfileStrategy} (Finalization): If the user wants to get metrics about a traversal's performance, then a terminal profile() step serves as a place holder for later inserting profile() steps between each step in the traversal. Profiling metrics are maintained in a side-effect data structure.
\begin{equation*}
\begin{array}{lcl}
\texttt{a.b.profile()} &\mapsto& \texttt{a.profile().b.profile()} 
\end{array}
\end{equation*}
\item \textbf{ComputerVerificationStrategy} (Verification): Single-machine Gremlin (OLTP) is more flexible in the types of traversals it can execute. When a traversal is executed within a multi-machine compute cluster (OLAP), certain traversal sequences are not allowed. Over time, as advances are made to the distributed Gremlin machine, respective verifications will be removed accordingly.
\begin{equation*}
\begin{array}{lcl}
\texttt{a.order.b} &\mapsto& \texttt{error} \\  
\texttt{a.local(out().out()).b} &\mapsto& \texttt{error} 
\end{array}
\end{equation*}
\end{itemize}

\section{Distributed Graph Traversals}

Gremlin's \textit{graph computer} traversal machine was designed to support the distributed execution of a Gremlin traversal via the bulk synchronous parallel (BSP) model of distributed computing \cite{bsp:valiant1990}. In graph-based BSP, each vertex is a logical processor that receives messages (typically via adjacent neighbors), updates its state given its current state and its received messages, and send messages (typically to adjacent neighbors) \cite{vertex:mccune2015}. This process continues until there are no more messages being sent. The result of the computation is distributed across the state of all the vertices (e.g. a \textit{score}-property on each vertex). In Gremlin, the vertices receive traversers as messages, execute the traverser's traversal step identified by $\psi$, and for each traverser generated, sends a message to the respective vertex referenced by the traverser's $\mu$.\footnote{If the traverser references an edge, then the message is sent to the vertex maintaining that edge. If the traverser references some other object in $U$, then the traverser remains at the current vertex location.} For those traversers that have halted, the vertex saves the halted traversers in a ``hidden" vertex property value. This process continues until there are no more traversers being sent around the cluster. The aggregate of the locations of all halted traversers across all vertices in the graph is the result of the computation. 

When a graph is partitioned across a cluster, traversers migrate between the machines as dictated by their $\mu(t)$.\footnote{It is not required that the Gremlin graph computer machine operate over multiple physical machines. The only requirement is that the underlying graph system provide a logical partition of the vertices in $V$ and that each vertex in each partition has direct (non-remote) references to its properties, incident edges, and incident edge properties.} For instance, $V_1 \uplus V_2 \uplus V_3 = V$ denote three vertex partitions of $V$, where each partition is composed of a unique set of vertices of $V$ and their respective incident edges ($V_1 \cap V_2 \cap V_3 = \emptyset$). When traverser $t$ is located at $\mu(t) \in V_1$, the traverser will exist at machine 1. Suppose the step $\psi(t)$ is applied and the following traverser's are generated $(t',t'',t''')$ where, $\mu(t') \in V_2$, $\mu(t'') \in V_3$, and $\mu(t''') \in V_1$. Traversers $t'$ and $t''$ will be serialized and sent over the network to machines 2 and 3, respectively, for further execution. However, traverser $t'''$ will remain at machine 1 to execute its referenced step $\psi(t''')$. In this model of distributed graph traversal, it is advantageous to create a partition of $V$ that reduces costly inter-machine communication.

The Gremlin graph computer machine executes in a breadth-first manner as each traverser at every vertex is operated on in parallel. As graphs become large, the number of traversers can easily grow exponentially, especially in repeated one-to-many mappings. For example, in a 20x20 lattice with vertices having only one ``right" and one ``down" edge, the traversal
\begin{verbatim}
  g.V(topLeft).repeat(out()).times(40)
\end{verbatim}
will ultimately yield $\sim$138 billion traversers at the lattice's ``bottom right" vertex. The general equation for the number of traversers for any $n$x$n$ lattice is $\binom{2n}{n} = \frac{(2n)!}{(n!)^2}$. While the number of traversers can grow exponentially, the number of traverser locations is always bounded by the size of $U$, where for this particular lattice traversal, the upper bound is $|V| = 400$. Gremlin \textit{bulking} takes advantage of this traverser-to-location relation by projecting the (potentially exponentially growing) traverser set to a traverser set constrained to an upper limit $|U|$. As a result of this \textit{lossless compression}, in the 20x20 lattice example, the resultant traverser set contains only one traverser $t$ whose bulk is $\beta(t) = 137,846,528,820$. Traversers in $T$ are ``bulked" according to their respective equivalence class
\begin{equation*}
\begin{array}{lrcrl}
[t] = \{ t' \in T \;\; | & \mu(t) &=& \mu(t') & \wedge \\
                          & \psi(t) &=& \psi(t') & \wedge \\
                          & \Delta(t) &=& \Delta(t') & \wedge \\
                          & \varsigma(t) &=& \varsigma(t') & \wedge \\
                          & \iota(t) &=& \iota(t') & \} .
\end{array}
\end{equation*}
In other words, all components of the traverser's 6-tuple must be equal amongst all traversers in $[t]$ except for their respective bulk.\footnote{Many traversals do not require the labeled path (or a full labeled path) of a traverser and as such, in practice, when the labeled path is not required, then $\Delta(t) = \emptyset$. This helps to ensure a smaller number of equivalence classes in $T$ and thus, a greater likelihood of bulking. Note that when labeled paths are considered, typically, the number of equivalence classes grows proportionate to the number of unique $\Psi$ paths in $G$ which tends to grow exponentially in broad, non-filtering traversals.} The traversers in $[t]$ are reduced to a single traverser $t''$, where $\beta(t'') = \Sigma_i{\beta\left([t]_i\right)}$ and $|[t'']| = 1$. Instead of enumerating each traverser, each traverser is counted (with only one traverser existing in $T$ for each equivalence class). Bulking ensures that breadth-first Gremlin on large graphs (OLAP) does not require an exponential amount of memory. Furthermore, this optimization is leveraged in depth-first Gremlin (OLTP) via the barrier() step (and \texttt{LazyBarrierStrategy}), where
\begin{equation*}
\text{barrier} : \left[U^*\right] \rightarrow U^* 
\end{equation*}
and 
\begin{equation*}
\text{barrier}(T) = \bigcup_{\forall [t] \in T} t_{\beta(t) = \sum_{i}{\beta\left([t]_i\right)}}.
\end{equation*}

\section{Minimal Gremlin Traversal Machines} 

Gremlin, as a graph traversal machine, is an automaton. Automata are studied in computer science to understand the minimal structural and behavioral requirements necessary for some ``abstract machine" to perform a particular type of computation if it were to be physically built or modeled in another machine powerful enough to simulate its behavior \cite{hopcroft:automata1979}. In automata theory, computations are represented as problems in language transduction. Every language is composed of an alphabet of characters in $\Sigma$. Strings in the language are in $\Sigma^*$. If an input string is ``1 + 2" and an automaton produces ``3," then the automaton has the requisite computing power to perform addition (assuming it generalizes to any ``$x + y$"). 

As it stands, the \textit{Turing machine} is the minimal abstract machine required for general-purpose computing. A Turning machine can do any known ``mechanical" (algorithmic) computation. Restricted forms of a Turning machine are able to solve simpler problems. For instance, a finite-state automaton can be programmed to process regular languages and answer regular expressions. For example, the finite-state automaton programmed as $a^*b$ maps $b$, $ab$, $aab$, $aaab$, etc. to \textbf{true}.

The types of languages that particular minimal Gremlin machines can process are provided in the table below, where each minimal Gremlin machine is a reduction of the elements in the traverser's 6-tuple structure and/or a reduction in the set of possible steps that can be used to program $\Psi$. For Gremlin to simulate a Turing machine (as well as any less powerful automata), its instruction set must at minimum support values(), property(), sack(), choose(), repeat(), in(), and out().
\begin{center}
\begin{tabular}{| l | l | l |}\hline
\textbf{Automaton} & \textbf{Minimal Gremlin Machine} & \textbf{Language}  \\ \hline
FiniteState & $(U_\mu \times \Psi \setminus \text{property} \cup \text{in} \times U_\varsigma)$ & Regular   \\ \hline
Pushdown & $(U_\mu \times \Psi \setminus \text{in} \times U_\varsigma)$ & ContextFree \\ \hline
Turing & $(U_\mu \times \Psi \times U_\varsigma)$ & Recursive \\ \hline
\end{tabular}
\end{center}

\subsection{Turing Completeness}

This subsection defines a surjective function whose image of the full Gremlin traversal machine is isomorphic to a single-headed Turing machine. The aforementioned components of Gremlin are mapped to the components of a 5-tuple Turing machine
\begin{equation*}
M = (Q, F, \Gamma, \Sigma, \delta),
\end{equation*}
where $Q$ is the set of machine states, $F \subseteq Q$ is the set of legal halt states, $\Gamma$ is the readable/writable alphabet, $\Sigma \subseteq \Gamma$ is the initial input on an infinite one-dimensional tape of cells, and $\delta: \left((Q \setminus F) \times \Gamma\right) \rightarrow \left(Q \times \Gamma \times \{L,R\}\right)$ is the transition function which updates the machine's state, determines which symbol to write to the current tape cell, and whether it should then go ``left" or ``right" on the tape. 
\begin{equation*}
\begin{array}{l}
\ldots[\;\;][\;\;][\gamma][\;\;][\;\;][\;\;]\ldots\\
\;\;\;\;\;\;\;\;\;\;\;\;\;\;\uparrow \\
\;\;\;\;\;\;\;\;\;\;\;\;\;M[\delta]   
\end{array}
\end{equation*}
A machine is deemed Turing Complete if it can simulate the aforementioned single-headed Turing machine. If so, that machine can be programmed to execute any known algorithm \cite{compute:turing1937}.
\begin{theorem}
The Gremlin graph traversal machine is Turing Complete.
\end{theorem}
\begin{proof}
For the Gremlin traversal machine, the tape is the (infinite) graph $G$, where each vertex $v \in V$ has one incoming neighbor, one outgoing neighbor, a \textit{symbol}-property value in $\Gamma$ and initially, the symbols $\Sigma$ are the \textit{symbol}-property values of a consecutive chain of vertices. In other words, $G$ is a line graph with each vertex representing a cell in the Turing tape with respective input symbols. The Turing machine state is the traverser's sack $\varsigma(t) \in Q$. Assume the existence of only the following Gremlin steps:
\begin{equation*}
\begin{array}{rlcrl}
\text{values}_\text{symbol} &: V \rightarrow \Gamma & & \text{map} & \text{read tape} \\
\text{property}_{\text{symbol} \in \Gamma} &: V \rightarrow V & & \text{sideEffect} & \text{write tape} \\ 
\text{sack} & : V \rightarrow Q & & \text{map} & \text{read state} \\
\text{sack}_{q \in Q} & : V \rightarrow V & & \text{sideEffect} & \text{write state} \\
\text{choose} &: V \rightarrow V & & \text{branch} & \text{if/else} \\
\text{repeat} &: V \rightarrow V & & \text{branch} & \text{loop} \\
\text{in} &: V \rightarrow V & & \text{map} & \text{move left} \\
\text{out} &: V \rightarrow V & & \text{map} & \text{move right}.
\end{array}
\end{equation*}
The $\delta$ function of the Turing machine is the composition of the steps above to create $\Psi$. Note that the steps above map to traverser sets of size 1 as no step is a true flatMap() nor filter(). Given that there is only one ``left" and one ``right" adjacent neighbor to a vertex in the line graph, the above steps will never increase (nor decrease) the size of the traverser set. In this way, any single-headed Turing machine can be simulated.
\end{proof}

As an example, a 3-state ``busy beaver" Turing machine is defined, where $\Gamma = \{0,1\}$, $Q = \{\text{A},\text{B},\text{C},\text{HALT}\}$, and $\delta$ is the Gremlin traversal below.
\begin{verbatim}
g.V(1).
 sack("A").
 repeat(choose(values("symbol")).
  option("0", choose(sack()).
   option("A", 
    property("symbol","1").out().sack("B")).
   option("B", 
    property("symbol","1").in().sack("A")).
   option("C", 
    property("symbol","1").in().sack("B"))).
  option("1", choose(sack()).
   option("A", 
    property("symbol","1").in().sack("C")).
   option("B", 
    property("symbol","1").out().sack("B")).
   option("C", 
    property("symbol","1").out().sack("HALT")))).
 until(sack().is("HALT"))
\end{verbatim}

\subsection{A Universal Gremlin Machine}

A Universal Gremlin Machine (UGM) is a Gremlin machine that can simulate another Gremlin machine within its $G$, $\Psi$ and $T$ constructs \cite{gpsemnet:rodriguez2010,rodriguez:rvm2011}. The encoding to follow will represent both $\Psi$ and $T$ in $G$ \cite{grammar:rodriguez2008,ripple:shinavier2007}. Any step in $\Psi$ can be represented as a vertex $v \in V$, where $\lambda(v,\text{label}) = \text{step}$ and $\lambda(v,\text{op})$ is the operation of that step (e.g. $\text{out}$). If $u \in V$ is another step in $\Psi$ that follows $v$ then there exists an edge $(v,u) \in E$, where $\lambda((v,u),\text{label}) = \text{nextStep}$.\footnote{Further complications exist for nested traversals. However, for the sake of brevity, the general representation of such traversals are left to the reader to contemplate.} A traverser in $T$ can be represented by a vertex $t \in V$ where, $\lambda(t,\text{label}) = \text{traverser}$ and each element of the traverser's 6-tuple is a subgraph in $G$.
\begin{enumerate}
 \item $\mu(t)$ is the edge $(t,x) \in E$, where $\lambda((t,x),\text{label}) = \text{mu}$.
 \item $\psi(t)$ is the edge $(t,y) \in E$, where $\lambda((t,y),\text{label}) = \text{psi}$ and $\lambda(y,\text{label}) = \text{step}$.
 \item $\Delta(t)$ is a collection of edges to vertices that join step labels (as vertex properties) with graph locations in $G$.
 \item $\beta(t)$ is the traverser vertex's bulk property $\lambda(t,\text{bulk}) = \beta(t)$.
 \item $\varsigma(t)$ is the traverser vertex's sack property $\lambda(t,\text{sack}) = \varsigma(t)$.
 \item $\iota(t)$ is the traverser vertex's loop counter property $\lambda(t,\text{loops}) = \iota(t)$.
\end{enumerate}
The aforementioned mapping represents both $\Psi$ and $T$ as subgraphs of $G$ such that, in total, $G$ contains the complete representation of the structure (graph) and the process (traversal and traversers) of the computation. However, representation is not execution. In order for this graph structure to evolve (and thus, compute), there must exist a Universal Gremlin Machine traverser and respective traversal that moves between $G \setminus \left(\Psi \cup T\right)$ and the traversal $\Psi \subset G$ updating the respective edges and properties of the traversers $T \subset G$. The $G$-encoded traversal and traversers form a \textit{virtual machine} in the Universal Gremlin Machine. A snippet of the $\Psi_\text{UGM}$ is presented below where all supported steps would need to be represented in an \texttt{option(...)}.

\begin{verbatim}
  g.V().hasLabel("traverser").as("t").
   repeat(
    choose(out("psi").values("op")).
     option("out", 
      outE("mu").as("drop").inV().out().
      addInE("mu","t"))
     option("in",
      outE("mu").as("drop").inV().in().
      addInE("mu","t"))
     option(...)
     option(...)
    sideEffect(select("drop").drop()).
    select("t").
    outE("psi").as("drop").inV().
    out("nextStep").addInE("psi","t").
    sideEffect(select("drop").drop()).
    select("t")).
   until(out("psi").count().is(0))
\end{verbatim}

The $\Psi_\text{UGM}$ traversal loops over its repeat() traversal until the $G$-encoded traverser halts by no longer referencing a step vertex. The result of the computation is the multi-set union of the symbols on the ``tape"-subgraph $G \setminus \left(\Psi \cup T\right)$. Formally, 
\begin{equation*}
\text{result} = \biguplus_i^{|V|}{ 
\begin{cases}
\lambda\left(V_i, \text{symbol}\right) &: \lambda\left(V_i,\text{label}\right) \notin \{ \text{traverser}, \text{step} \} \\ 
\emptyset &: \text{otherwise}.
\end{cases}}
\end{equation*}

In order to provide a Universal Gremlin Machine that can operate on $G$-encoded Gremlin machines that maintain the same level of expressivity as the Gremlin traversal machine discussed in this article, it would be necessary to extend the above $\Psi_\text{UGM}$ traversal to account for the growing and shrinking of $G$-encoded traverser sets as well as all the steps of Gremlin's instruction set.

\subsection{Parallel Universal and $G$-Encoded Machines}

It has been assumed, up to this point, that the traversers in $T$ all reference steps of the same $\Psi$. However, nothing prevents multiple traverser sets to exist, where each set operates under a different traversal. In fact, regardless of $G$-encoded machines, this is necessary for allowing parallel, concurrent traversals/queries of $G$. With respect to $G$-encoded machines, the Universal Gremlin Machine need not concern itself with which traverser of which traversal it is executing. In fact, the Universal Gremlin Machine simply needs to find any traverser that has yet to halt and execute its next step. The Universal Gremlin Machine acts as a \textit{thread} evolving the state of different traversals/programs. However, in order to get a well defined result set for each traversal, a $\Psi$-unique identifier would need to be appended to each traverser so that the result of some traversal $\Psi^{123}$ can be unambiguously gathered via
\begin{equation*}
\text{result}_{\Psi^{123}} = \biguplus_i^{|V|}{
\begin{cases}
\mu(V_i) &: \lambda\left(V_i,\text{label}\right) = \text{traverser} \;\; \wedge \\ & \;\;\; \lambda\left(V_i,\text{psiId}\right) = 123 \;\; \wedge \\ & \;\;\;\;\; \lambda\left(V_i,\text{psi}\right) = \emptyset \\
\emptyset &: \text{otherwise}.
\end{cases}}
\end{equation*}
Finally, nothing prevents multiple Universal Gremlin Machines operating in parallel against $G$ locating active traversers and executing a step until no more traversers exist or all traversers have halted. This is, in fact, analogous to a multi-threaded system.

\subsection{Traversing a Gremlin Traversal Machine}

When the graph $G$, the traversal $\Psi$, and the traverser set $T$ are all encoded in $G$, then all the components of a Gremlin traversal machine exist in the same address space -- namely $G$. A consequence of this co-location is that a traverser can, in principle, traverser its own structure. Similarly, a traverser can traverse its traversal. When a machine has direct reference to its representation, a machine can not only analyze itself via \textit{reflection}, but it can also rewrite itself. The ramifications of this consequence, with respects to applied graph computing, are left to future ruminations.

\subsection{A Primordial Graph Traversal Machine}

This section describes, at a high-level, a vision of graph computing that is, in many ways, analogous to the token rewrite model of the lambda calculus \cite{lambda:chruch1936}. A lossless, injective function takes a multi-relational, attributed digraph (\textbf{MADG}) to a multi-relational, unattributed digraph (\textbf{MDG}), where edges are reified structures and all properties are ``property key"-labeled edges incident to ``property value" vertices \cite{rdfproperty:hartig2014}. Next, there exists an injective function that maps a multi-relational digraph to an unlabeled digraph (\textbf{DG}), where labels are encoded as ``binary vertex chains" \cite{mapnetwork:rodriguez2009}. Finally, another injective function has been defined that maps a digraph to an undirected graph (\textbf{UG}), where edge directions are represented as topological features of the undirected form \cite{mapnetwork:rodriguez2009}. 
\begin{equation*}
\textbf{MADG} \mapsto \textbf{MDG} \mapsto \textbf{DG} \mapsto \textbf{UG}
\end{equation*}
Given the existence of this mapping, the complete state of computing (i.e. $G$, $\Psi$, and $T$) can be represented by a single undirected graph whose structure is solely the composition of ``dots and lines" in some $n$-dimensional space. In this primitive, verbose graph, there are no labels, strings, numbers, etc., simply dots connected to each other by lines. Computing occurs when subgraphs of a particular shape (e.g. a traverser at a location in the graph) morph to form new subgraphs of a particular shape (e.g. new traversers with new graph locations). Computing, in this manner, can be conceptualized as a chemical reaction where ``molecular structures" (undirected subgraphs) interact with adjacent structures to yield new structures that may elicit yet more reactions \cite{cot:dittrich2007}. In this primordial world, the computation is complete when vertices and edges are no longer being created nor destroyed. When the undirected graph reaches an equilibrium with its ``laws of physics," the problem is solved -- for it has reached a stable state.

\section{Conclusion}

Gremlin is a graph traversal machine and language. The Gremlin machine specification is simple to describe and ultimately implement. The complexity of the computations that Gremlin enables is not necessarily due to its constructs, but due to the data sets being processed. Graphs are multi-dimensional structures able to model a heterogenous set of ``things" related to each other in a heterogenous set of ways -- all within a single, connected data structure. When a Gremlin traversal is evaluated against a graph, billions upon billions of traversers can be generated on even small graphs due to the exponential growth of the number of paths that exist with each step the traversers take. With so many forks in the road, traversers continually split themselves in order to explore each option that meets the constraints of the traversal they obey. When these traversers ultimately halt, they provide an answer to the question specified by their traversal, which was programmed by a user via the Gremlin traversal language. 

\acks
The Apache TinkerPop project (\texttt{http://tinkerpop.com}) was started in November 2009 and is currently in its third generation of development with TinkerPop3 having been released in July of 2015. Many individuals have contributed to the project and wider ecosystem over the years and their contributions, both theoretical and applied, have been invaluable to the generation of the ideas presented in this article.

\bibliography{../marko}
\bibliographystyle{abbrvnat}

\end{document}